\documentclass[envcountsame,envcountsect,10pt,oribibl]{llncs}

\usepackage{microtype}

\usepackage[usenames,dvipsnames]{xcolor}

\usepackage{amssymb,amstext,amsfonts,amsopn,stmaryrd}

\usepackage{tabularx}
\usepackage{array,multirow} 

\usepackage{comment}

\usepackage{graphics}
\usepackage{amsmath}
\usepackage{subfig}			 
\usepackage{stackrel}
\usepackage{tikz}
\usepackage{mathrsfs}
\usepackage{enumerate}
\usepackage{amssymb}

\graphicspath{{./}}

\newcommand{\pname}[1]{\textnormal{\textsc{#1}}}
\newcommand{\cclass}[1]{\textnormal{\textsf{#1}}}

\newcommand{\HED}{\pname{$H$-free Edge Deletion}}
\newcommand{\HEC}{\pname{$H$-free Edge Completion}}
\newcommand{\HEE}{\pname{$H$-free Edge Editing}}

\newcommand{\HDED}{\pname{$H'$-free Edge Deletion}}

\newcommand{\HDEE}{\pname{$H'$-free Edge Editing}}

\newcommand{\HBEC}{\pname{$\overline{H}$-free Edge Completion}}
\newcommand{\HBEE}{\pname{$\overline{H}$-free Edge Editing}}

\newcommand{\TDDED}{\pname{$t$-diamond-free Edge Deletion}}

\newcommand{\TODDED}{\pname{$(t-1)$-diamond-free Edge Deletion}}

\newcommand{\RED}{\pname{$R$-free Edge Deletion}}
\newcommand{\REE}{\pname{$R$-free Edge Editing}}

\newcommand{\PTED}{\pname{$P_3$-free Edge Deletion}}
\newcommand{\PTEE}{\pname{$P_3$-free Edge Editing}}

\newcommand{\PFEE}{\pname{$P_4$-free Edge Editing}}

\newcommand{\CLEE}{\pname{$C_\ell$-free Edge Editing}}

\newcommand{\TWKTEE}{\pname{$2K_2$-free Edge Editing}}

\newcommand{\DED}{\pname{Diamond-free Edge Deletion}}
\newcommand{\DEE}{\pname{Diamond-free Edge Editing}}

\newcommand{\HDBEC}{\pname{$\overline{H'}$-free Edge Completion}}

\newcommand{\TSAT}{\pname{3-SAT}}

\newcommand{\NP}{\cclass{NP}}
\newcommand{\NPC}{\cclass{NP-complete}}

\newcommand{\PH}{$\cclass{NP} \subseteq \cclass{coNP/poly}$}

\newcommand{\edges}[1][V]{\left[#1\right]^2} 

\newtheorem{construction}{Construction}

\newtheorem{observation}[lemma]{Observation}

\frontmatter          
\newcommand{\defstage}[2]{
  \hfill\\\smallskip\noindent%
  \begin{tabularx}{\textwidth}{|l X|}%
    \hline%
    \multicolumn{2}{|l|}{\textbf{#1}}\\%
    &#2\\\hline%
  \end{tabularx}%
}

\begin{document}
\mainmatter              
\title{Parameterized Lower Bounds and Dichotomy Results for the NP-completeness of $H$-free Edge Modification Problems}
%
%
%
%
\author{N. R. Aravind\inst{1} \and R. B. Sandeep\inst{1}\thanks{supported by TCS Research Scholarship} \and Naveen Sivadasan\inst{2}}
\institute{Department of Computer Science \& Engineering\\
Indian Institute of Technology Hyderabad, India\\
\email{\{aravind,cs12p0001\}\makeatletter@\makeatother iith.ac.in}
\and
TCS Innovation Labs, Hyderabad, India\\
\email{naveen\makeatletter@\makeatother atc.tcs.com}}

\maketitle              
\begin{abstract}
For a graph $H$, the \HED\ problem asks whether there exist at most $k$ 
edges whose deletion from the input graph $G$ 
results in a graph without any induced copy of $H$. 
\HEC\ and \HEE\ are defined similarly where
only completion (addition) of edges are allowed in the former and 
both completion and deletion are allowed in the latter. 
We completely settle the classical complexities of these problems by proving that
\HED\ is \NPC\ if and only if $H$ is a graph with at least two edges, 
\HEC\ is \NPC\ if and only if $H$ is a graph with at least two non-edges and
\HEE\ is \NPC\ if and only if $H$ is a graph with at least three vertices.
Additionally, we prove that, these \NPC\ problems cannot be 
solved in parameterized subexponential time, i.e., 
in time $2^{o(k)}\cdot |G|^{O(1)}$, unless Exponential Time Hypothesis fails.
Furthermore, we obtain implications on the incompressibility of these 
problems.
\end{abstract}
\section{Introduction}
\label{sec:introduction}

Edge modification problems are to test whether modifying at most $k$ edges
makes the input graph satisfy certain properties. The three major edge modification
problems are edge deletion, edge completion and edge editing problems.
In edge deletion problems we are allowed to delete at most $k$ edges from
the input graph. Similarly, in completion problems, it is allowed to
complete (add) at most $k$ edges and in editing problems at most
$k$ editing (deletion or completion) are allowed. Edge modification problems
comes under the broader category of graph modification problems which have
found applications in DNA physical mapping 
\cite{goldberg1995four},
numerical algebra \cite{rose1972graph}, circuit design \cite{el1988complexity}
and machine learning \cite{DBLP:journals/ml/BansalBC04}.

The focus of this paper is on $H$-free edge modification problems,
in which we are allowed to modify at most $k$ edges to make the input
graph devoid of any induced copy of $H$, where $H$ is any fixed graph.
Though these problems have been studied for four decades, a complete
dichotomy result on the classical complexities of these problems are not yet found.
We settle this by proving that \HED\ is \NPC\ if and only if $H$ is a graph with at least two edges, 
\HEC\ is \NPC\ if and only if $H$ is a graph with at least two non-edges and
 \HEE\ is \NPC\ if and only if $H$ is a graph with at least three vertices.
As a bonus, we obtain the parameterized lower bounds for these \NPC\ problems.
We obtain that these \NPC\ problems cannot be solved in parameterized 
subexponential time (i.e., in time $2^{o(k)}\cdot |G|^{O(1)}$), unless
Exponential Time Hypothesis (ETH) fails, where ETH is a widely believed complexity theoretic
assumption. Furthermore, we obtain implications on the incompressibility 
(non-existence of polynomial kernels) of these problems.

We build on our recent paper \cite{DBLP:conf/cocoa/AravindSS15},
in which we proved that \HED\ is \NPC\ if $H$ has at least two edges and 
has a component with
maximum number of vertices which is a tree or a regular graph.
We also proved that these problems cannot be solved in parameterized 
subexponential time, unless ETH fails.

\paragraph{\textbf{Related Work:}}
In 1981, Yannakakis proved that \HED\ is \NPC\
if $H$ is a cycle~\cite{DBLP:journals/siamcomp/Yannakakis81}. 
Later in 1988, El-Mallah and Colbourn proved that the problem
is \NPC\ if $H$ is a path of at least two edges \cite{el1988complexity}.
Addressing the fixed parameter tractability of a generalized version of these problems,
Cai proved that \cite{DBLP:journals/ipl/Cai96}~\HED, \textsc{Completion} and \textsc{Editing}
are fixed parameter tractable, i.e., they can be solved in time
$f(k)\cdot |G|^{O(1)}$, for some function $f$. 
Polynomial kernelizability of these problems
have been studied widely. 
Given an instance $(G,k)$ of the problem the objective is to 
obtain in polynomial time an equivalent instance of size
polynomial in $k$. Kratsch and Wahlstr{\"o}m gave the first result on the incompressibility of 
$H$-free edge modification problems. They proved that \cite{kratsch2009two} for a certain graph $H$
on seven vertices, \HED\ and \HEE\ do not admit polynomial kernels, unless \PH. They use
polynomial parameter transformation from an \NPC\ problem and hence their results imply the NP-completeness
of these problems. Later, Cai and Cai proved that 
\HEE, \textsc{Deletion} and \textsc{Completion} do not admit polynomial kernels
if $H$ is a path or a cycle with at least four edges, unless \PH~\cite{DBLP:journals/algorithmica/CaiC15}.
Further, they proved that \HEE\ and \textsc{Deletion} are incompressible if $H$ is 3-connected but not complete, 
and \HEC\ is incompressible if $H$ is 3-connected and has at least two non-edges,
unless \PH~\cite{DBLP:journals/algorithmica/CaiC15}.
Under the same assumption, 
it is proved that \HED\ and \HEC\ are incompressible if $H$ is a tree on at least 7 vertices, which is not a star graph and
\HED\ is incompressible if $H$ is the star graph $K_{1,s}$, where $s\geq 10$~\cite{cai2012polynomial}.
They also use polynomial parameter transformations and hence these problems
are \NPC. 


\paragraph{\textbf{Outline of the Paper:}}
Section~\ref{sec:prebasics} gives the notations and terminology used in the paper.
It also introduces a construction which is a modified version of the main
construction used in \cite{DBLP:conf/cocoa/AravindSS15}.
Section~\ref{sec:editing} settles the case of \HEE.
Section~\ref{sec:deletion} obtains results for \HED\ and \textsc{Completion}.
In the concluding section, we discuss the implications of our results on the 
incompressibility of $H$-free edge modification problems. 
%
\section{Preliminaries and Basic Tools}
\label{sec:prebasics}

\paragraph{\textbf{Graphs:}}

For a graph $G$, $V(G)$ denotes the vertex set and 
$E(G)$ denotes the edge set. We denote the symmetric
difference operator by $\triangle$, i.e., for two sets $F$ and $F'$,
$F\triangle F'=(F\setminus F') \cup (F'\setminus F)$.
For a graph $G$ and a set $F\subseteq \edges[V(G)]$, $G\triangle F$
denotes the graph $(V(G),E(G)\triangle F)$. A component of a graph is largest
if it has maximum number of vertices. By $|G|$ we denote $|V(G)|+|E(G)|$.
The disjoint union of two graphs $G$ and $G'$ is denoted by $G\cup G'$
and the disjoint union of $t$ copies of $G$ is denoted by $tG$.
A simple path on $t$ vertices is denoted by $P_t$.
The graph $t$-diamond is $K_2+tK_1$, the join of $K_2$ and $tK_1$. 
Hence, $2$-diamond is the diamond graph.
The minimum degree of a graph $G$ is denoted by $\delta(G)$
and the maximum degree is denoted by $\Delta(G)$.
Degree of a vertex $v$ in a graph $G$ is denoted by $\deg_G(v)$.
We remove the subscript when there is no ambiguity.
We denote the complement of a graph $G$ by $\overline{G}$.
For a graph $H$ and a vertex set $V'\subseteq V(H)$,
$H[V']$ is the graph induced by $V'$ in $H$.
A null graph is a graph without any edge.

For integers $\ell$ and $h$ such that $h>\ell$, $(\ell,h)$-degree graph is a graph in which 
every vertex has degree either $\ell$ or $h$. The set of vertices with degree $\ell$
is denoted by $V_{\ell}$ and the set of vertices with degree $h$ is denoted by $V_h$. 
An $(\ell,h)$-degree graph is called \emph{sparse} if $V_l$ induces a graph with at most 
one edge and $V_h$ induces a graph with at most one edge.

The context determines whether \HED\ denotes the classical problem or the parameterized
problem. This applies to \textsc{Completion} and \textsc{Editing} problems.
For the parameterized problems, 
we use $k$ (the size of the solution being sought) as the parameter.
In this paper, edge modification implies either deletion, completion or editing.

\paragraph{\textbf{Technique for Proving Parameterized Lower Bounds:}}
Exponential Time Hypothesis (ETH) is a widely believed complexity theoretic
assumption that \TSAT\ 
cannot be solved in time $2^{o(n)}$, where 
$n$ is the number of variables in
the \TSAT\ instance. 
A linear parameterized reduction is a polynomial time reduction
from a parameterized problem $A$ to a parameterized problem $B$
such that for every instance $(G,k)$ of $A$, the reduction gives
an instance $(G',k')$ such that $k'=O(k)$. The following
result helps us to obtain parameterized lower bound under ETH.

\begin{proposition}[\cite{DBLP:books/sp/CyganFKLMPPS15}]
  \label{pro:lpr}
  If there is a linear parameterized reduction from a parameterized problem $A$
  to a parameterized problem $B$ and if $A$ does not admit a parameterized subexponential
  time algorithm, then $B$ does not admit a parameterized subexponential time algorithm.
\end{proposition}

Two parameterized problems $A$ and $B$ are linear parameter equivalent if there is a linear 
parameterized reduction from $A$ to $B$ and there is a linear parameterized reduction from 
$B$ to $A$. We refer the book \cite{DBLP:books/sp/CyganFKLMPPS15} for various aspects of 
parameterized algorithms and complexity.
The following are some folklore observations.

\begin{proposition}
  \label{pro:folklore}
  \HED\ and \HBEC\ are linear parameter equivalent. Similarly,
  \HEE\ and \HBEE\ are linear parameter equivalent.
\end{proposition}

\begin{proposition}
  \label{pro:equivalence}
  \begin{enumerate}[(i)]
  \item\label{item:dc-equivalence} \HED\ is \NPC\ if and only if \HBEC\ is \NPC.
    Furthermore, \HED\ cannot be solved in parameterized subexponential time
    if and only if \HBEC\ cannot be solved in parameterized subexponential time.
  \item\label{item:ee-equivalence} \HEE\ is \NPC\ if and only if \HBEE\ is \NPC.
    Furthermore, \HEE\ cannot be solved in parameterized subexponential time
    if and only if \HBEE\ cannot be solved in parameterized subexponential time.
  \end{enumerate}
\end{proposition}

\begin{proposition}
  \label{pro:polynomial}
  \begin{enumerate}[(i)]
  \item\label{item:poly-deletion} \HED\ is polynomial time solvable if $H$ is a graph with
    at most one edge.
  \item\label{item:poly-completion} \HEC\ is polynomial time solvable if $H$ is a graph with
    at most one non-edge.
  \item\label{item:poly-editing} \HEE\ is polynomial time solvable if $H$ is a graph with
  at most two vertices.
  \end{enumerate}
\end{proposition}

In this paper, we prove that these are the only polynomial time solvable 
$H$-free edge modification problems.
For any fixed graph $H$, the $H$-free edge modification problems trivially belong to \NP. Hence, we may 
state that these problems are \NPC\ by proving their NP-hardness.
\subsection{Basic Tools}

The following construction is a slightly modified version of the main 
construction used in \cite{DBLP:conf/cocoa/AravindSS15}.
The modification is done to make it work for reductions of \textsc{Completion}
and \textsc{Editing} problems. 
The input of the construction is a tuple $(G',k,H,V')$, where $G'$ and $H$ are
graphs, $k$ is a positive integer and $V'\subseteq V(H)$. In the old construction
(Construction 1 in \cite{DBLP:conf/cocoa/AravindSS15}), for every copy $C$ of $H[V']$
in $G'$, we introduced $k+1$ copies of $H$ such that the intersection of every pair
of them is $C$. In the modified construction given below, we do the same for every copy $C$
of $H[V']$ on a complete graph on $V(G')$. 


\begin{construction}
  \label{con:nonadj}
  Let $(G',k, H, V')$ be an input to the construction, where $G'$ and $H$ are graphs, $k$
  is a positive integer and $V'$ is a subset of vertices of $H$.
  Label the vertices of $H$ such that every vertex gets a unique label. Let the labelling be $\ell_H$.
  Consider a complete graph $K'$ on $V(G')$. 
  For every subgraph (not necessarily induced) $C$ with a vertex set $V(C)$ 
  and an edge set $E(C)$ in $K'$ such that $C$ is isomorphic to $H[V']$,
  do the following:
  \begin{itemize}
    \item Give a labelling $\ell_C$ for the vertices in $C$ such that there is an isomorphism
      $f$ between $C$ and $H[V']$ which maps every vertex $v$ in $C$ to a vertex $v'$ in $H[V']$
      such that $\ell_C(v)=\ell_H(v')$, i.e., $f(v)=v'$ if and only if $\ell_C(v)=\ell_H(v')$.
    \item Introduce $k+1$ sets of vertices $V_1, V_2,\ldots, V_{k+1}$, each of size $|V(H)\setminus V'|$.
    \item For each set $V_i$, introduce an edge set $E_i$ of size $|E(H)\setminus E(H[V'])|$ among
      $V_i\cup V(C)$
       such that there is an isomorphism $h$ between $H$ and
      $(V(C)\cup V_i, E(C)\cup E_i)$ which preserves $f$, i.e.,
      for every vertex $v\in V(C)$, $h(v)=f(v)$.
  \end{itemize}
  This completes the construction. Let the constructed graph be $G$. 
\end{construction}

We remark that the complete graph $K'$ on $V(G')$ is not part of the constructed graph.
The complete graph is only used to find where we need to introduce new vertices and edges.
An example of the construction is shown in Figure~\ref{fig:cons}.
We use the terminology used in \cite{DBLP:conf/cocoa/AravindSS15}. We repeat it here for convenience.
Let $C$ be a copy of $H[V']$ in $K'$. Then, $C$ is called a \emph{base}.
Let $\{V_i\}$ be the $k+1$ sets of vertices introduced in the construction for the base $C$.
Then, each $V_i$ is called a \emph{branch} of $C$ and the vertices in $V_i$ are called 
the \emph{branch vertices} of $C$. If $V_j$ is a branch of $C$, then
the vertex set of $C$ is denoted by $B_j$.
The vertex set of $G'$ in $G$ is denoted by $V_{G'}$. The copy of $H$ formed by $V_j$, $E_j$
and $C$ is denoted by $H_j$. Since $H$ is a fixed graph, the construction runs in polynomial time. 
The following two Lemmas are the generalized version of Lemma 2.3 and 3.5 of \cite{DBLP:conf/cocoa/AravindSS15}.

\begin{figure}[h]
  \centering
  \subfloat[Subfigure 1 list of figures text][$G'$]{
    \includegraphics[width=1.5in]{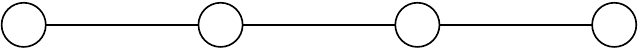}
    \label{fig:gdash}}
  \qquad
  \subfloat[Subfigure 2 list of figures text][$H$. The vertices in $V'$ are blackened.]{
    \includegraphics[width=1.0in]{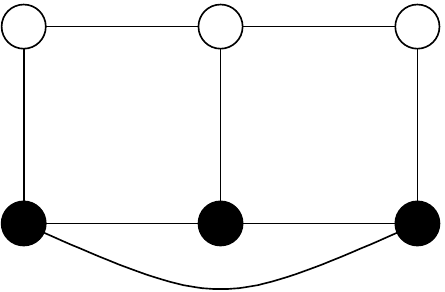}
    \label{fig:h}}
  \qquad
  \subfloat[Subfigure 2 list of figures text][Output of Construction~\ref{con:nonadj} with an input $(G',k=1,H,V')$.]{
    \includegraphics[width=1.5in]{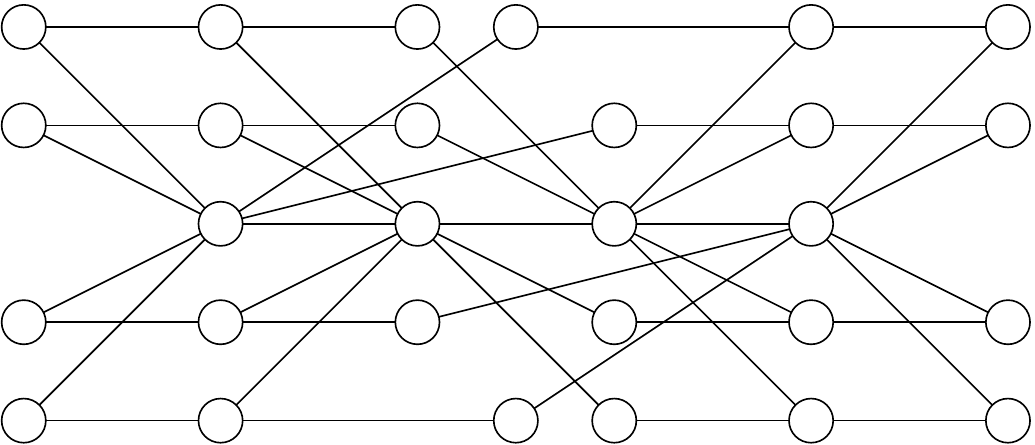}
    \label{fig:g}}
  \caption{An example of Construction~\ref{con:nonadj}}
  \label{fig:cons}
\end{figure}

\begin{lemma}
  \label{lem:con:nonadj-backward}
  Let $G$ be obtained by Construction~\ref{con:nonadj} on 
  the input $(G',k,H,V')$, where $G'$ and $H$ are graphs, $k$ is a positive integer and $V'\subseteq V(H)$.
  Then, if $(G,k)$ is a yes-instance of \HEE~(\textsc{Deletion}/\textsc{Completion}), 
  then $(G',k)$ is a yes-instance of \HDEE~(\textsc{Deletion}/\textsc{Completion}), 
  where $H'$ is $H[V']$.
\end{lemma}
\begin{proof}
  Let $F$ be a solution of size at most $k$ of $(G,k)$. For a contradiction, assume that $G'\triangle F$
  has an induced $H'$ with a vertex set $U$. Hence there is a base $C$ in $G'$ isomorphic to
  $H'$ with the vertex set $V(C)=U$. Since there are $k+1$ copies of $H$ in $G$, where each pair
  of copies of $H$ has the intersection $C$, and $|F|\leq k$, operating with $F$ cannot kill all the copies of 
  $H$ associated with $C$. Therefore,
  since $U$ induces an $H'$ in $G'\triangle F$, there exists a branch $V_i$ of $C$ such that $U\cup V_i$
  induces $H$ in $G\triangle F$, which is a contradiction.\qed
\end{proof}


\begin{lemma}
  \label{lem:degree}
  Let $H$ be any graph and $d$ be any integer. Let $V'$ be the set of vertices
  in $H$ with degree more than $d$.
  Let $H'$ be $H[V']$. Then, there is a linear parameterized reduction
  from \HDEE~(\textsc{Deletion}/\textsc{Completion}) to \HEE~(\textsc{Deletion}/\textsc{Completion}).
\end{lemma}
\begin{proof}
  Let $(G',k)$ be an instance of \HDEE~(\textsc{Deletion}/\textsc{Completion}).
  Apply Construction~\ref{con:nonadj} on $(G',k,H,V')$ to obtain $G$.
  We claim that $(G',k)$ is a yes-instance of \HDEE~(\textsc{Deletion}/\textsc{Completion})
  if and only if $(G,k)$ is a yes-instance of \HEE~(\textsc{Deletion}/\textsc{Completion}).

  Let $F'$ be a solution of size at most $k$ of $(G',k)$. For a contradiction,
  assume that $G\triangle F'$ has an induced $H$ with a vertex set $U$.
  Since a branch vertex has degree at most $d$, every vertex in $U$
  with degree more than $d$ in $(G\triangle F')[U]$ must be from $V_{G'}$.
  Hence there is an induced $H'$ in $G'\triangle F'$, which is a contradiction.
  Lemma~\ref{lem:con:nonadj-backward} proves the converse.\qed
\end{proof}
\section{\HEE}
\label{sec:editing}

In this section, we prove that \HEE\ is \NPC\
if and only if $H$ is a graph with at least three vertices.
We also prove that these problems cannot be solved in 
parameterized subexponential time unless ETH fails.
We use the following known results.

\begin{proposition}
  \label{pro:editing-base}
  The following problems are \NPC. Furthermore, they cannot be 
  solved in time $2^{o(k)}\cdot |G|^{O(1)}$, unless ETH fails.
  \begin{enumerate}[(i)]
  \item\label{item:editing-p3} \PTEE~\cite{komusiewicz2012cluster}.
  \item\label{item:editing-p4} \PFEE~[Follows from the proof of the lower 
    bound of \textsc{$\{C_4,P_4\}$-free Edge Editing} in \cite        
    {drange2015trivially}\footnote{We thank P{\aa}l Gr{\o}n{\aa}s Drange for pointing out this and sharing a complete proof of the same.}].
  \item\label{item:editing-cycle} \CLEE, for any fixed $l\geq 3$~[Follows
    from the proof for the corresponding \textsc{Deletion} problems in \cite{DBLP:journals/siamcomp/Yannakakis81}].
  \item\label{item:editing-2k2} \TWKTEE~[(\ref{item:editing-cycle}) and 
    Proposition~\ref{pro:equivalence}(\ref{item:ee-equivalence})].
  \item\label{item:editing-diamond} \DEE~\cite{DBLP:journals/classification/BarthelemyB01}.
  \end{enumerate}
\end{proposition}

In our previous work \cite{DBLP:conf/cocoa/AravindSS15},
we proved that \RED\ is \NPC\ if $R$ is a regular graph
with at least two edges. We also proved that these
\NPC\ problems cannot be solved in parameterized 
subexponential time, unless ETH fails.
We observe that the results for \RED\ follows for
\REE\ as well. The proofs are very similar except
that we use Construction~\ref{con:nonadj} instead of its
ancestor in \cite{DBLP:conf/cocoa/AravindSS15} and we
reduce from \textsc{Editing} problems instead of \textsc{Deletion}
problems. We can use \PTEE, \CLEE\ and \TWKTEE\ as the base cases instead of their
\textsc{Deletion} counterparts.
We skip the proof as it will 
be a repetition of that in \cite{DBLP:conf/cocoa/AravindSS15}.

\begin{lemma}
  \label{lem:editing-reg-temp}
  Let $R$ be a regular graph with at least two edges.
  Then \REE\ is \NPC. Furthermore, the problem cannot be
  solved in time $2^{o(k)}\cdot |G|^{O(1)}$, unless ETH fails.
\end{lemma}

Now, we strengthen the above lemma by proving the same results for all 
regular graphs with at least three vertices.

\begin{lemma}
  \label{lem:editing-reg}
  Let $R$ be a regular graph with at least three vertices.
  Then \REE\ is \NPC. Furthermore, the problem cannot be
  solved in time $2^{o(k)}\cdot |G|^{O(1)}$, unless ETH fails.
\end{lemma}
\begin{proof}
  If $R$ has at least two edges then the statements follows from 
  Lemma~\ref{lem:editing-reg-temp}. Assume that
  $R$ has at most one edge and at least three vertices. It is 
  straight-forward to see that $R$ must be the null graph.
  Then the complement of $R$ is a complete graph with at least two edges. 
  Now, the statements follows from Proposition~\ref{pro:equivalence}(\ref{item:ee-equivalence})
  and Lemma~\ref{lem:editing-reg-temp}.\qed
\end{proof}

Having these results in hand, we use Lemma~\ref{lem:degree} to prove the 
dichotomy result and the parameterized lower bound of \HEE.
Given a graph $H$ with at least three vertices, we introduce a method 
Editing-Churn($H$) to 
obtain a graph $H'$ such that there is a linear parameterized reduction
from \HDEE\ to \HEE\ and $H'$ is a graph with at least three vertices
and is a regular graph or a $P_3$ or a $P_4$ or a diamond.

\defstage{Editing-Churn($H$)}
{ $H$ is a graph with at least three vertices.
  \begin{enumerate}[Step 1:]
    \item\label{item:churn-ed1} If $H$ is a regular graph, a $P_3$, a $P_4$ or a diamond, then return $H$.
    \item\label{item:churn-ed2} If $H$ is a graph in which the number of vertices with degree more than $\delta(H)$
      is at most two, then let $H=\overline{H}$ and goto Step~\ref{item:churn-ed1}.
    \item\label{item:churn-ed3} Delete all vertices with degree $\delta(H)$ in $H$ and go to Step~\ref{item:churn-ed1}.
  \end{enumerate}
}

\begin{observation}
  \label{obs:editing-churn}
  Let $H$ be a graph with at least three vertices.
  Then Editing-Churn($H$) returns a graph $H'$ which has
  at least three vertices and is a regular graph or a $P_3$ or a $P_4$
  or a diamond. Furthermore, there is a linear parameterized reduction
  from \HDEE\ to \HEE.
\end{observation}
\begin{proof}
  At any stage of the method, we make sure that the graph has at least three vertices. 
  Let $H'$ be an intermediate graph obtained in the method such that it is neither a regular graph
  nor a $P_3$ nor a $P_4$ nor a diamond. If Step 2 is applicable to both $H'$ and $\overline{H'}$,
  then $H$ hat at most four vertices. Hence $H$ has either three or four vertices. 
  It is straight-forward to verify that a graph (with three or four vertices) or its complement,
  satisfying the condition in Step 2,
  is either a regular graph or a $P_3$ or a $P_4$ or a diamond, which is a contradiction.
  The linear parameterized reduction from \HDEE\ to \HEE\ follows from
  Proposition~\ref{pro:equivalence}(\ref{item:ee-equivalence}) and Lemma~\ref{lem:degree}.\qed
\end{proof}

\begin{theorem}
  \label{thm:editing}
  \HEE\ is \NPC\ if and only if $H$ is a graph with 
  at least three vertices. Furthermore, these \NPC\
  problems cannot be solved in time $2^{o(k)}\cdot |G|^{O(1)}$, unless ETH fails.
\end{theorem}
\begin{proof}
  If $H$ is a graph with at most two vertices, 
  the statements follows from Proposition~\ref{pro:polynomial}(\ref{item:poly-editing}).
  Let $H$ be a graph with at least three vertices.
  Let $H'$ be the graph returned by Editing-Churn($H$).
  By Observation~\ref{obs:editing-churn}, $H'$ is either a regular graph or a $P_3$
  or a $P_4$ or a diamond and there is a linear parameterized reduction from
  \HDEE\ to \HEE.
  Now, the statements follows from the lower bound results for these graphs 
  (\ref{pro:editing-base}(\ref{item:editing-p3}), (\ref{item:editing-p4}), (\ref{item:editing-diamond}) and
  Lemma~\ref{lem:editing-reg}).\qed
\end{proof}
\section{\HED}
\label{sec:deletion}

In this section, we prove that
\HED\ is \NPC\ if and only if $H$ is a graph with 
at least two edges. We also prove that
these \NPC\ problems cannot be solved in parameterized subexponential time,
unless ETH fails. Then, from Proposition~\ref{pro:equivalence}(\ref{item:dc-equivalence}), we
obtain a dichotomy result for \HEC. We apply a technique similar to that we 
applied for \textsc{Editing} in the last section. 

\begin{proposition}
  \label{pro:deletion-base}
  The following problems are \NPC. Furthermore, they cannot be 
  solved in time $2^{o(k)}\cdot |G|^{O(1)}$, unless ETH fails.
  \begin{enumerate}[(i)]
  \item\label{item:deletion-p3} \PTED~\cite{komusiewicz2012cluster}.
  \item\label{item:deletion-diamond} \DED~\cite{DBLP:journals/disopt/FellowsGKNU11,DBLP:conf/ipec/SandeepS15}.
  \item\label{item:deletion-tree-reg} \HED, if $H$ is a graph with at least two edges and
    has a largest component which is a regular graph or a tree~\cite{DBLP:conf/cocoa/AravindSS15}.
  \end{enumerate}
\end{proposition}

The following Lemma is a consequence of Lemma~\ref{lem:degree} and Proposition~\ref{pro:equivalence}(\ref{item:dc-equivalence}).

\begin{lemma}
  \label{lem:rotate}
  Let $H$ be any graph. Then the following hold true:
  \begin{enumerate}[(i)]
  \item Let $H'$ be the subgraph of $H$ obtained by removing all
    vertices with degree $\delta(H)$.
    Then there is a linear parameterized reduction from \HDED\ to \HED.
  \item Let $H'$ be the subgraph of $H$ obtained by removing all
    vertices with degree $\Delta(H)$.
    Then there is a linear parameterized reduction from \HDED\ to \HED. 
  \end{enumerate}
\end{lemma}
\begin{proof}
  The first part directly follows from Lemma~\ref{lem:degree} by setting $d=\delta(H)$.
  To prove the second part, consider the problem \HBEC. 
  Let $H''$ be the graph obtained by removing all vertices with degree $\delta(\overline{H})$ from $\overline{H}$. 
  Now, by Lemma~\ref{lem:degree},
  there is a linear parameterized reduction from \textsc{$H''$-free Edge Completion} to \HBEC.
  We observe that $H''$ is $\overline{H'}$. Hence, by Proposition~\ref{pro:equivalence}(\ref{item:dc-equivalence}),
  there is a linear parameterized reduction from \HDED\ to \HED.\qed
\end{proof}

Given a graph $H$, we keep on deleting either the minimum degree vertices or the 
maximum degree vertices by making sure that the resultant graph has at least 
two edges. We do this process until we obtain a graph
in which vertices with degree more than $\delta(H)$ induces a graph with 
at most one edge and vertices with degree less than $\Delta(H)$ induces
a graph with at most one edge. We call this method Deletion-Churn.

\defstage{Deletion-Churn($H$)}
{ $H$ is a graph with at least two edges.
  \begin{enumerate}[Step 1:]
    \item\label{item:churn-del1} If $H$ is a graph in which the vertices with degree more than $\delta(H)$
      induces a subgraph with at most one edge and the vertices with degree less than $\Delta(H)$ induces
      a subgraph with at most one edge, then return $H$.
    \item\label{item:churn-del2} If $H$ is a graph in which the vertices with degree more than $\delta(H)$
      induces a subgraph with at least two edges, then delete all vertices with degree $\delta(H)$ from $H$ and goto Step 1.
    \item\label{item:churn-compl3} If $H$ is a graph in which the vertices with degree less than $\Delta(H)$
      induces a subgraph with at least two edges, then delete all vertices with degree $\Delta(H)$ from $H$. Goto Step 1.
  \end{enumerate}
}

\begin{observation}
  \label{obs:deletion-churn}
  Let $H$ be a graph with at least two edges. If the vertices with degree more than $\delta(H)$ 
  induces a graph with at most one edge and the vertices with degree less than $\Delta(H)$ induces
  a graph with at most one edge, then $H$ is either regular graph or a forest or a sparse $(\ell,h)$-degree graph.
\end{observation}
\begin{proof}
  Assume that $H$ is not a regular graph.
  Since $H$ has at least two edges and it satisfies the premises, $\delta(H)\geq 1$.
  If $\delta(H)=1$, the premises imply that $H$ is a forest. Assume that $\delta(H)\geq 2$.
  Then we prove that $H$ is a sparse $(\ell,h)$-degree graph.
  For a contradiction, assume that there exists a vertex $v\in V(H)$ such that $\delta(H)<deg(v)<\Delta(H)$.
  The premises imply that $v$ has degree at most two, which is a contradiction.\qed
\end{proof}

\begin{lemma}
  \label{lem:deletion-churn}
  Let $H$ be a graph with at least two edges. Then Deletion-Churn($H$)
  returns a graph $H'$ such that:
  \begin{enumerate}[(i)]
  \item\label{item:deletion-churn-red} There is a linear parameterized reduction from \HDED\ to \HED.
  \item\label{item:deletion-churn-output} $H'$ has at least two edges and is either a regular graph
    or a forest or a sparse $(\ell, h)$-degree graph.
  \end{enumerate}
\end{lemma}
\begin{proof}
  In every step, we make sure that there are at least two edges in the resultant graph.
  Now, the first part follows from Lemma~\ref{lem:rotate} and the second part follows
  from Observation~\ref{obs:deletion-churn}.\qed
\end{proof}

If the output of Deletion-Churn($H$), $H'$ is a regular graph or a forest,
we obtain from Proposition~\ref{pro:deletion-base}(\ref{item:deletion-tree-reg}) that
\HED\ is \NPC\ and cannot be solved in parameterized subexponential
time, unless ETH fails.
Therefore, the only graphs to be handled now is the sparse $(\ell,h)$-degree graphs with 
at least two edges. We do that in the next two subsections.
\subsection{\TDDED}
\label{sec:tdiamond}

We recall that $t$-diamond is the graph $K_2+tK_1$ and that 
2-diamond is the diamond graph. Clearly, $t$-diamond is a sparse $(\ell, h)$-degree graph.
In this subsection, we prove that \TDDED\ is \NPC. Further, we 
prove that the problem cannot be solved in parameterized subexponential time,
unless ETH fails. We use an inductive proof where the base case is \DED.
For the proof, we introduce a simple construction, which is 
given below.

\begin{figure}[h]
  \centering
  \includegraphics[width=1.0in]{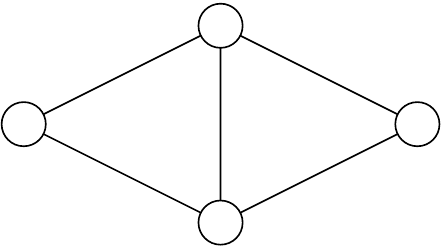}
  \caption{A 2-diamond is isomorphic to a diamond graph.}
  \label{fig:diamond}
\end{figure}

\begin{construction}
  \label{con:diamond}
  Let $(G',k)$ be an input to the construction.
  For every edge $\{u,v\}$ in $G'$, introduced a clique $C_{\{u,v\}}$ of $k+1$
  vertices such that every vertex in $C_{\{u,v\}}$ is adjacent to both $u$ and $v$.
  This completes the construction. Let $G$ be the resultant graph.
\end{construction}

\begin{lemma}
  \label{lem:tdiamond}
  For any $t\geq 2$, \TDDED\ is \NPC. Furthermore,
  the problem cannot be solved in time $2^{o(k)}\cdot |G|^{O(1)}$, unless
  ETH fails.
\end{lemma}
\begin{proof}
  The proof is by induction on $t$. If $t=2$, the problem is \DED\ and the 
  theorem follows from Proposition~\ref{pro:deletion-base}(\ref{item:deletion-diamond}). 
  Assume that $t\geq 3$ and that the statements hold true for $t-1$.
  We give a reduction from \TODDED\ to \TDDED. 

  Let $(G',k)$ be an instance of \TODDED. Apply Construction~\ref{con:diamond}
  on $(G',k)$ to obtain $G$. We claim that $(G',k)$ is a yes-instance of \TODDED\
  if and only if $(G,k)$ is a yes-instance of \TDDED. 

  Let $(G',k)$ be a yes-instance of \TODDED. Let $F'$ be a solution of size 
  at most $k$ of $(G',k)$. We claim that $F'$ is a solution of $(G,k)$.
  For a contradiction, assume that $G-F'$ has an induced $t$-diamond on a 
  vertex set $U\subseteq V(G)$. Let $x$ and $y$ be the $(t+1)$-degree vertices in 
  the $t$-diamond induced by $U$ in $G-F'$. Now there are three cases to be considered.
  
  Case 1: Both $x$ and $y$ are 
  from a clique $C_{\{u,v\}}$ introduced in the construction. 

  We note that $x$ and $y$
  are adjacent to $u$ and $v$ and all other vertices in $C_{\{u,v\}}$. Hence the common neighborhood
  of $x$ and $y$ does not have an independent set of size at least $3$, which is a contradiction.

  Case 2: Let $x$ is from a clique $C_{\{u,v\}}$ introduced in the construction and $y$ be $u$.

  The common neighborhood
  of $x$ and $y$ does not have an independent set of size at least $3$, which is a contradiction.

  Case 3: Both $x$ and $y$ are from $G'$. The common neighborhood of $x$ and $y$ in $G-F'$
  is constituted by $C_{\{x,y\}}$ and the common neighbors of $x$ a and $y$ in $G'-F'$.
  Since $C_{\{x,y\}}$ is a clique, it can contribute at most one to the independent set of the 
  common neighborhood of $x$ and $y$. Hence, there should be an independent set of size at least $t-1$
  in the common neighborhood of $x$ and $y$ in $G'-F'$. Since $G'-F'$ is $(t-1)$-diamond-free, this is a contradiction.

  Conversely, let $F$ be a solution of size at most $k$ of $(G,k)$. We prove that $G'-F$ is $(t-1)$-diamond-free.
  For a contradiction, assume that $G'-F$ has an induced $(t-1)$-diamond on a vertex set $U\subseteq V(G')$.
  Let $x$ and $y$ be the $t$-degree vertices of the $(t-1)$-diamond induced by $U$ in $G'-F$.
  Since there are $k+1$ common neighbors of $x$ and $y$ ($C_{\{x,y\}}$) introduced by the construction, 
  there exists a common neighbor $z\in C_{\{x,y\}}$ such that $U\cup \{z\}$
  induces a $t$-diamond in $G-F'$, which is a contradiction.\qed
\end{proof}

\subsection{Handling sparse $(\ell,h)$-degree graphs}
\label{sec:sparselh}

We recall that for $h>\ell$, 
every vertex of a sparse $(\ell, h)$-degree graph $H$ is either of degree $\ell$
or of degree $h$ and that $V_{\ell}$ induces a graph with at most one edge and $V_h$
induces a graph with at most one edge. We have already handled $t$-diamond graphs.
We handle the rest of the sparse $(\ell,h)$-degree graphs in this subsection.
Let $H$ be any sparse $(\ell,h)$-graph.
There are four cases to be handled:
\begin{description}
\item[Case 1:] $V_{h}$ is an independent set; $V_{\ell}$ is an independent set
\item[Case 2:] $V_{h}$ induces a graph with one edge; $V_{\ell}$ is an independent set
\item[Case 3:] $V_{h}$ is an independent set; $V_{\ell}$ induces a graph with one edge
\item[Case 4:] $V_{h}$ induces a graph with one edge; $V_{\ell}$ induces a graph with one edge
\end{description}

\begin{observation}
  \label{obs:sparselh}
  Let $H$ be a sparse $(\ell,h)$-graph with at least two edges. Then the following hold true:
  \begin{enumerate}[(i)]
  \item\label{item:sparselh-1} If $\ell=1$, then $H$ is a forest.
  \item\label{item:sparselh-2} If $\ell\geq 2$, then $|V_{\ell}|\geq 2$ and the equality holds only when $H$ is a diamond.
  \end{enumerate}
\end{observation}
\begin{proof}
  To prove the first part, we observe that $H\setminus V_{\ell}$
  has at most one edge. 
  To prove the second part, we observe that if $|V_{\ell}|\leq 2$ and if $H$ is not a diamond, then $h\leq \ell$,
  which is a contradiction.\qed
\end{proof}

Since the case of forest is already handled in Proposition~\ref{pro:deletion-base}(\ref{item:deletion-tree-reg}),
we can safely assume that $\ell\geq 2$ and hence $h\geq 3$.
We start with handling Case 1.
We use a slightly modified 
version of Construction~\ref{con:nonadj}. 
We recall that, in Construction~\ref{con:nonadj}, with an input $(G',k,H,V')$,
For every copy $C$ of $H[V']$ in $K'$ (a complete graph on $V(G')$), we introduced $k+1$ branches such that
each branch along with $C$ form a copy of $H$. 
In the modified
construction, in addition to this, we make every pair of vertices from different branches mutually adjacent.

\begin{construction}
  \label{con:adj}
  Let $(G',k, H, V')$ be an input to the construction, where $G'$ and $H$ are graphs, $k$
  is a positive integer and $V'$ is a subset of vertices of $H$.
  Apply Construction~\ref{con:nonadj} on $(G',k,H,V')$ to obtain $G''$.
  For every pair of vertices $\{v_i,v_j\}$ such that $v_i\in V_i$ and $v_j\in V_j$, where $i\neq j$,
  make $v_i$ and $v_j$ adjacent.
  This completes the construction. Let the constructed graph be $G$. 
\end{construction}

Now, we have a lemma similar to Lemma~\ref{lem:con:nonadj-backward}. 
We skip the proof as it is quite similar to that of Lemma~\ref{lem:con:nonadj-backward}. 

\begin{lemma}
  \label{lem:con:adj-backward}
  Let $G$ be obtained by Construction~\ref{con:adj} on 
  the input $(G',k,H,V')$, where $G'$ and $H$ are graphs, $k$ is a positive integer and $V'\subseteq V(H)$.
  Then, if $(G,k)$ is a yes-instance of \HED, then $(G',k)$ is a yes-instance of \HDED, where $H'$ is $H[V']$.
\end{lemma}

\begin{lemma}
  \label{lem:lhbipartite}
  Let $H$ be a sparse $(\ell,h)$-graph, where $h>\ell\geq 2$ such that both $V_{\ell}$ and $V_h$
  are independent sets.
  Then \HED\ is \NPC. Furthermore, the problem cannot be solved in time
  $2^{o(k)}\cdot |G|^{O(k)}$, unless ETH fails.
\end{lemma}
\begin{proof}
  We reduce from \PTED. Let $V'=\{u,v,w\}\subseteq V(H)$ be such that
  $v\in V_h$, $u,w\in V_{\ell}$ and $V'$ induces a $P_3$ in $H$.
  Since $h\geq 3$, such a subset of vertices does exist in $H$.
  Let $(G',k)$ be an instance of \PTED.
  Apply Construction~\ref{con:adj} on $(G',k,H,V')$
  to obtain $G$. Let $H'$ be $H[V']$. 
  We claim that $(G',k)$ is a yes-instance of \PTED\ if and only if 
  $(G,k)$ is a yes-instance of \HED.

  Let $(G',k)$ be a yes-instance of \PTED. Let $F'$ be a solution
  of size at most $k$ of $(G',k)$. For a contradiction, assume that
  $G-F'$ has an induced $H$ on a vertex set $U$. Let $V^U_{\ell}$ and
  $V^U_{h}$ be the $V_{\ell}$ and $V_h$ respectively of the $H$ induced by
  $U$ in $G-F'$.

  Claim 1: $V^U_{h}$ is a subset of a single branch, say $V_1$.

  Since $V^U_h$
  is an independent set in $(G-F')[U]$, 
  $V^U_{h}$ cannot span over multiple branches. Hence $V^U_{h}\subseteq V_1\cup V_{G'}$.
  Let $x\in V^U_h\cap V_{G'}$. Consider the neighborhood of $x$, $N(x)$ in $(G-F')[U]$.
  Since the neighborhood of every vertex in $H$ is triangle-free, $N(x)$ cannot 
  contain vertices from multiple branches. Further, since $G'-F'$ is $P_3$-free, 
  $N(x)$ can have at most one vertex from $V_{G'}$. Let $x$ is adjacent to vertices in $V_1$. 
  We note that, by construction, 
  $x$ has at most $h-2$ neighbors in $V_1$. Therefore $|N(x)|<h$, which is a contradiction.
  Thus we obtained that $V^U_h\subseteq V_1$.

  Claim 2: $|V^U_{\ell}\cap V_{G'}|\leq 1$

  Assume that $x\in U\cap V_{G'}$.
  Since degree of $x$ in $(G-F')[U]$ is $\ell$, $x$ must have $\ell$
  edges to $V^U_h$. Therefore, $x$ must be the middle vertex of the 
  $P_3$ formed by $B_1$ in $G'$. 

  Claim 1 and 2 imply that $|U\cap (V_1\cup V_{G'})| \leq |U|-2$. Hence,
  there exists a branch, other than
  $V_1$, say $V_2$, such that $V^U_{\ell}\cap V_2\neq \emptyset$. Since $V_{\ell}$ is an independent set,
  no other branches can have vertices in $V^U_{\ell}$. Therefore,
  $V^U_{\ell}\subseteq V_2\cup \{x\}$. Let $y$ be a vertex in $V^U_{\ell}\cap V_2$.
  Since $y$ is adjacent to all vertices in $V^U_{h}$, $\ell=|V^U_h|$.
  Hence $H$ is a complete bipartite graph. Further, $|V^U_{\ell}\cap V_2|\geq |V_{\ell}|-1$.
  It is straight-forward to verify that $V_2$ does not have an independent set of size 
  $|V_{\ell}|-1$, which is a contradiction.
  Lemma~\ref{lem:con:adj-backward} proves the converse.\qed
\end{proof}

Now we handle the cases in which $V_{\ell}$ induces a graph with one edge. 

\begin{lemma}
  \label{lem:sparselh-vl}
  Let $H$ be a sparse $(\ell,h)$-graph with at least two edges 
  such that $V_l$ induces a graph with one edge.
  Let $v_{\ell_1}$ and $v_{\ell_2}$ be the two adjacent vertices in $V_{\ell}$.
  Let $H'$ be the graph induced by $V(H)\setminus \{v_{\ell_1}, v_{\ell_2}\}$.
  Then, there is a linear parameterized reduction from \HDED\ to \HED.
\end{lemma}
\begin{proof}
  Let $(G',k)$ be an instance of \HDED. Apply Construction~\ref{con:nonadj}
  on $(G',k,H,V')$, where $V'$ is $V(H)\setminus \{v_{\ell_1}, v_{\ell_2}\}$.
  Let $G$ be the graph obtained from the construction. We claim that 
  $(G',k)$ is a yes-instance of \HDED\ if and only if $(G,k)$ is a 
  yes-instance of \HED.

  Let $(G',k)$ be a yes-instance of \HDED\ and let $F'$ be a solution of size
  at most $k$ of $(G',k)$. For a contradiction, assume that $G-F'$ has an induced
  $H$ with a vertex set $U$. It is straight-forward to verify that 
  If a branch vertex $v_1\in V_1$ is in $U$, then its neighbor in the same branch $u_1\in V_1$
  must be in $U$ and both acts as $v_{\ell_1}$ and $v_{\ell_2}$ in the $H$
  induced by $U$ in $G-F'$. Hence $'-F'$ has an induced $H'$, which is a contradiction.
  Lemma~\ref{lem:con:nonadj-backward} proves the converse.\qed
\end{proof}

\begin{observation}
  \label{obs:sparselh-vl}
  Let $H$ be a sparse $(\ell,h)$-graph with at least two edges where $h>\ell\geq 2$
  such that $V_l$ induces a graph with one edge.
  Let $v_{\ell_1}$ and $v_{\ell_2}$ be the two adjacent vertices in $V_{\ell}$.
  Let $H'$ be the graph induced by $V(H)\setminus \{v_{\ell_1}, v_{\ell_2}\}$.
  Then $H'$ has at least two edges.
\end{observation}
\begin{proof}
  By Observation~\ref{obs:sparselh}(\ref{item:sparselh-2}), since $H$ is not a diamond,
  $|V_{\ell}|\geq 3$. This implies that $V\setminus \{v_{\ell_1}, v_{\ell_2}\}$ is 
  nonempty. Now the observation follows from the fact that $\ell\geq 2$.\qed
\end{proof}

Now we handle Case 2, i.e., $V_h$ induces a graph with one edge and $V_{\ell}$
is an independent set.

\begin{lemma}
  \label{lem:vh1vl0}
  Let $H$ be a sparse $(\ell,h)$ graph where $h>\ell\geq 2$, $V_h$
  induces a graph with one edge and $V_{\ell}$ is an independent set.
  Let $H$ be not a $t$-diamond.
  Let $v_{h_1}$ and $v_{h_2}$ be the two adjacent vertices in $H[V_h]$.
  Let $V'$ be $V_{\ell}\cup \{v_{h_1},v_{h_2}\}$. Let $H'$ be $H[V']$. 
  Then, there is a linear
  parameterized reduction from \HDED\ to \HED.
\end{lemma}
\begin{proof}
  For convenience, we give a reduction from \HDBEC\ to \HBEC. Then the statements follow 
  from Proposition~\ref{pro:equivalence}(\ref{item:dc-equivalence}).

  Let $(G',k)$ be an instance of \HDBEC. Apply Construction~\ref{con:nonadj}
  on $(G',k,H,V')$, where $V'$ is $V_{\ell}\cup \{v_{h_1}, v_{h_2}\}$.
  Let $G$ be the graph obtained from the construction. We claim that 
  $(G',k)$ is a yes-instance of \HDBEC\ if and only if $(G,k)$ is a 
  yes-instance of \HBEC.

  Let $(G',k)$ be a yes-instance of \HDBEC\ and let $F'$ be a solution of size
  at most $k$ of $(G',k)$. For a contradiction, assume that $G+F'$ has an induced
  $H$ with a vertex set $U$. It is straight-forward to verify that 
  If a branch vertex $v_1\in V_1$ is in $U$, then all its neighbors in the same branch are in $U$ and
  $V_1$ acts as $V_h\setminus \{v_{h_1},v_{h_2}\}$ of $H$ in $\overline{H}$
  induced by $U$ in $G+F'$. Hence $G'+F'$ has an induced $\overline{H'}$, which is a contradiction.
  Lemma~\ref{lem:con:nonadj-backward} proves the converse.\qed  
\end{proof}

\begin{observation}
  \label{obs:vh1vl0}
  Let $H$ be a sparse $(\ell,h)$ graph where $h>\ell\geq 2$, $V_h$
  induces a graph with one edge and $V_{\ell}$ is an independent set.
  Let $H$ be not a $t$-diamond, for $t\geq 2$.
  Let $v_{h_1}$ and $v_{h_2}$ be the two adjacent vertices in $H[V_h]$.
  Let $V'$ be $V_{\ell}\cup \{v_{h_1},v_{h_2}\}$. Let $H'$ be $H[V']$. 
  Then $H'$ has at least two edges and $|V(H')| < |V(H)|$.
\end{observation}
\begin{proof}
  Follows from the facts that $h\geq 3$ and $H$ is not a $t$-diamond.\qed
\end{proof}

\begin{lemma}
  \label{lem:sparselh}
  Let $H$ be a sparse $(\ell,h)$-degree graph with at least two edges.
  Then \HED\ is \NPC. Furthermore, the problem cannot be solved in time
  $2^{o(k)}\cdot |G|^{O(1)}$, unless ETH fails.
\end{lemma}
\begin{proof}
  If $V_{\ell}$ induces a graph with an edge, then we apply
  the technique used in Lemma~\ref{lem:sparselh-vl} and obtain
  a graph $H'$ with at least two edges. Similarly, if $H$ is 
  not a $t$-diamond and $V_h$ induces a graph with an edge,
  then we apply the technique used in Lemma~\ref{lem:vh1vl0} to obtain
  a graph $H'$ with at least two edges. If the obtained graph $H'$
  is not a sparse $(\ell,h)$-degree graph, then we apply 
  Deletion-Churn($H'$) to obtain $H''$. We repeat this process until no more
  repetition is possible. Then, it is straight-forward to verify that
  we obtain a graph which is either a $t$-diamond, or a graph handled in Lemma~\ref{lem:lhbipartite} or
  a regular graph or a forest with at least two edges.\qed
\end{proof}
%
\subsection{Dichotomy Results}

We are ready to state the dichotomy results and the 
parameterized lower bounds for \HED\ and \HEC.

\begin{theorem}
  \label{thm:final}
  \HED\ is \NPC\ if and only if $H$ is a graph with at least two edges.
  Furthermore, the problem cannot be solved in time $2^{o(k)}\cdot |G|^{O(k)}$.
  \HEC\ is \NPC\ if and only if $H$ is a graph with at least two non-edges.
  Furthermore, the problem cannot be solved in time $2^{o(k)}\cdot |G|^{O(k)}$.
\end{theorem}
\begin{proof}
  Consider \HED. The statements follow from 
  Proposition~\ref{pro:polynomial}(\ref{item:poly-deletion}), Lemma~\ref{lem:deletion-churn}, 
  Proposition~\ref{pro:deletion-base}(\ref{item:deletion-tree-reg}) and Lemma~\ref{lem:sparselh}.
  Now the results for \HEC\ follows from Proposition~\ref{pro:equivalence}(\ref{item:dc-equivalence}).\qed
\end{proof}
\section{Concluding Remarks}

Our results have wide implications on the incompressibility of 
$H$-free edge modification problems. Polynomial parameter transformation (PPT)
is a widely used technique to prove the incompressibility of 
problems. To prove the incompressibility of a problem
it is enough to to give a PPT from a problem which is already
known to be incompressible, under some complexity theoretic assumption.
All our reductions are
linear parameterized reductions and hence are polynomial
parameter transformations. The following lemma is a
direct consequence of Lemma~\ref{lem:degree}.

\begin{lemma}
  \label{lem:incompressibility}
  Let $H$ be a graph and $d$ be any integer. 
  Let $H'$ be obtained from $H$ by deleting vertices with degree $d$ or less.
  Then, if \HDEE~(\textsc{Deletion}/\textsc{Completion})
  is incompressible, then \HEE~(\textsc{Deletion}/\textsc{Completion}) is incompressible.
\end{lemma}

We give a simple example to show an implication of this lemma. Consider an 
$n$-sunlet graph which is a graph in which 
a vertex with degree one is attached to each vertex of a cycle of $n$ vertices.
From the incompressibility of \textsc{$C_n$-free Edge Editing}, \textsc{Deletion} and \textsc{Completion},
for any $n\geq 4$, it follows that \textsc{$n$-sunlet-free Edge Editing}, \textsc{Deletion} and \textsc{Completion}
are incompressible for any $n\geq 4$.

We believe that our result is a step towards a dichotomy result on the incompressibility of $H$-free
edge modification problems. 
Another direction is to get a dichotomy result on the complexities of $\mathcal{H}$-free edge modification
problems where $\mathcal{H}$ is a finite set of graphs. 

%

\bibliographystyle{plain}
\bibliography{npc}
%

\end{document}